\documentclass[a4paper,10pt]{article}
\usepackage[utf8x]{inputenc}
\pdfoutput=1
\usepackage{nccfoots}

\usepackage{times}
\usepackage{geometry}
\usepackage[T1]{fontenc}
\usepackage{color}
\usepackage{amsmath}
\usepackage{amssymb}
\usepackage{array}
\usepackage{amsthm}
\usepackage{graphicx}
\usepackage{hyperref}
\usepackage{setspace}
\usepackage{stmaryrd}
\usepackage{marginnote}
\usepackage[mathscr]{euscript}

\theoremstyle{plain}
\newtheorem{thm}{Theorem}[section]
\newtheorem{prop}[thm]{Proposition}
\newtheorem{cor}[thm]{Corollary}

\theoremstyle{definition}

\theoremstyle{remark}
\newtheorem{rem}[thm]{Remark}

\theoremstyle{plain}

\newcommand{\R}{\mathbb{R}}

\newcommand{\eq}[1]{\begin{equation}#1\end{equation}}

\newcommand{\alg}[1]{\begin{aligned}#1\end{aligned}}

\newcommand{\bc}{\begin{cases}}
\newcommand{\ec}{\end{cases}}

\newcommand{\la}{\lambda}
\newcommand{\Si}{\Sigma}
\newcommand{\La}{\Lambda}

\renewcommand{\title}[1]{{\bfseries #1}\par}
\renewcommand{\author}[1]{\medskip{#1}\par\smallskip}

\numberwithin{equation}{section}

\begin{document}
\begin{center}
\title{\LARGE The Einstein--$\La$ flow on product manifolds}
\vspace{3mm}
\author{\large David Fajman and Klaus Kröncke}
\Footnotetext{}{2010 \emph{Mathematics Subject Classification.} 53B30,53C25,83C05.}
\Footnotetext{}{\emph{Key words and phrases.} General Relativity, Einstein metrics, Recollapse, positive cosmological constant.}

\vspace{1mm}
%\affiliation{Universität Potsdam, Institut für Mathematik\\Am Neuen Palais 10\\14469 Potsdam, Germany} 
%\email{klaus.kroencke@uni-potsdam.de} 
\end{center}

\begin{center}
July, 2016
\end{center}
\vspace{2mm}

\begin{abstract}
We consider the vacuum Einstein flow with a positive cosmological constant $\La$ on spatial manifolds of product form $M=M_1\times M_2$. 
In dimensions $n=\dim M\geq 4$ we show the existence of continuous families of recollapsing models whenever at least one of the factors $M_1$ or $M_2$ admits a Riemannian Einstein metric with positive Einstein constant. We moreover show that these families belong to larger continuous families with models that have two complete time directions, i.e.~do not recollapse. Complementarily, we show that whenever no factor has positive curvature, then any model in the product class expands in one time direction and collapses in the other. In particular, positive curvature of one factor is a necessary criterion for recollapse within this class. Finally, we relate our results to the instability of the Nariai solution in three spatial dimensions and point out why a similar construction of recollapsing models in that dimension fails.
The present results imply that there exist different classes of initial data which exhibit fundamentally different types of long-time behavior under the Einstein--$\La$ flow whenever the spatial dimension is strictly larger than three. Moreover, this behavior is related to the spatial topology through the existence of Riemannian Einstein metrics of positive curvature.
\end{abstract}
 
%

%\tableofcontents

\section{Introduction}
A fundamental open problem in the mathematical study of cosmological models is the determination of the long-time behavior of the Einstein flow ($\La=0$) or the Einstein-$\La$ flow ($\La>0$) on closed spatial manifolds $M$, where $\La$ denotes the cosmological constant. By \emph{Einstein flow} we refer to the form of the Einstein equations when a foliation of spacetime is chosen. In that case the equations appear as a system of PDEs that determines the evolution of geometric tensor fields and is in this way similar to other geometric flows. This term has been coined in earlier works (cf.~\cite{AnMo11}).  In the $\La=0$ case it is well-known that the long-time behavior depends on the topology of $M$. If $M$ is three-dimensional and of positive Yamabe type $Y_+$, it is conjectured that its maximal globally hyperbolic development (MGHD) recollapses, i.e.~is future and past geodesically incomplete \cite{BaGaTi86} (\emph{Closed Universe Recollapse conjecture}). This behavior is proven for the class of spherically symmetric solutions \cite{BuRe96,He02}. In the general case, the problem is open. In the negative Yamabe class $Y_-$ it is known that in an open neighborhood of initial data corresponding to the Milne model, this solution serves as an attractor of the Einstein flow \cite{AnMo11}. It is conjectured that under additional restrictions the global geometry of the Milne model (future completeness) is prototypical for the behavior of the Einstein flow in the classes $Y_-$ or $Y_0$ (but the spacetime has to be non-flat) in three spatial dimensions, i.e.~that the spacetimes admit CMC foliations where the mean curvature takes values in $(-\infty,0)$ and that this foliation covers the entire spacetime \cite{Re96}. (Cf.~also \cite{An01} for a discussion of the recollapse conjecture.)\\
For the Einstein-$\La$ flow with $\La>0$ on closed manifolds of dimension $n=3$ (i.e.~$n$ denotes the spatial dimension) the explicit solutions and their stability analysis suggest that the behavior of the flow is less topology-dependent and in particular recollapsing models are to our knowledge not known to exist (in the vacuum case). All known models possess at least one complete time direction, in the case of deSitter spacetime, which has spatial topology $\mathbb S^3$, both time directions are complete -- one of which can be viewn as future contracting, the other one as future expanding. These regions merge at a minimal surface. For Cauchy hypersurfaces in the negative or zero Yamabe class homogeneous models are past incomplete and future complete. All those models are future nonlinear stable \cite{Ri08}. In particular in an open neighborhood of initial data induced by those background solutions, the Einstein-$\La$ flow is well understood towards the expanding direction.

\subsection{The Einstein-$\La$ flow on product manifolds}
In this paper we prove that for any spatial dimension $n\geq 4$ there exist open sets of initial data, whose MGHDs recollapse under the vacuum Einstein-$\La$ flow. This implies in particular that the behavior of the Einstein-$\La$ flow in three spatial dimensions differs substantially from the situation in \emph{any} higher dimension. Moreover, we show that within the class of product manifolds, which we consider, positive curvature of at least one factor is necessary for recollapse. This will be made precise in the following.\\

\noindent The explicit models, which we construct are of product form
\eq{\label{ansatz}
\left(\overline M,\mathbf g\right)=\left(I\times M_1\times M_2,-dt^2+a(t)^2 g_1+b(t)^2g_2\right),
}
where $g_1$ and $g_2$ are Einstein metrics, i.e.~$\mathrm {Ric}[g_i]=\la_i g_i$ for $i=1,2$ and constants $\la_i$. We introduce $x=\log a$ and $y=\log b$. Initial data for spacetimes of this form at $t=0$ is given by $(M_1\times M_2,\la_1,\la_2,\dot x(0),\dot y (0))$. We assume $x(0)=y(0)=0$ without loss of generality. Note that other values for $x$ and $y$ can be encoded in the constants $\la_1$ and $\la_2$.
\subsection{Recollapsing models}
We show that if for both factors $M_i$, $\la_i$ is positive, then there exists a continuous 1-parameter family of initial data such that this family decomposes into two open sets, with recollapsing and non-recollapsing MGHDs, respectively, and one initial data set which necessarily corresponds to an unstable spacetime. An explicit condition for the family of initial data reveals that a necessary criterion for recollapse is a certain \emph{disequilibrium} between the two factors. We make this precise in the following theorem.

\begin{thm}\label{thm-1}
Let $\Lambda=\frac{n(n-1)}{2}$ and let $M_1$, $M_2$ be closed manifolds of  dimension $\dim M_i=n_i\geq 2$, $n_1+n_2=n$. Consider initial data of the form $x(0)=y(0)=\dot x(0)=\dot y(0)=0$ such that
\eq{\label{eq-def-rel}
n(n-1)=n_1\la_1+n_2\la_2,
}
holds. Then the following cases occur.
\begin{enumerate}

\item If $\la_2>n$ and $\la_1<(n_1-1)n/n_1$, the MGHD recollapses, i.e.~is future- and past incomplete, $I=(t_-,t_+)$ for $t_-<0<t_+$. Both singularities are of Kasner-type, i.e.~one factor collapses and the other factor has unbounded volume as $t\rightarrow t_{\pm}$. The mean curvature $H\rightarrow \pm\infty$ as $t\rightarrow t_{\pm}$. The Kretschmann scalar is unbounded towards both singularities. In particular, the MGHDs are $C^2$-inextendible.

\item If $\la_2=n$ and $\la_1=(n_1-1)n/n_1$, MGHD is 
\eq{\label{crit-sol}
\left(\mathbb R\times M_1\times M_2,-dt^2+\cosh\left(\sqrt{\frac{n}{n_1}}t\right)^2\cdot g_1 + g_2\right),
}
and in particular future and past complete.
\item If $n(n_1-1)/n_1<\la_1<(n-1)$ or equivalently $n-1<\la_2<n$ the MGHD is future- and past complete. In both time directions, both factors expand exponentially with asymptotic growth rates depending only on the dimensions $n_1$ and $n_2$. 
\end{enumerate}
Furthermore, the behavior is stable in the following sense. Initial data induced by any recollapsing solution in the class above has an open\footnote{Here, open is meant in the sense of $(g,K)\in H^\ell\times H^{\ell-1}$ spaces with $\ell=\ell(n)$ sufficiently large, where $(g,K)$ denote the spatial metric and second fundamental form, respectively.} neighborhood in the initial-data manifold, such that the MGHD of each element in this neighborhood recollapses. Similarly, initial data induced by any complete solution (except the critical solution) has an open neighborhood with elements that give rise to complete MGHDs. The critical solution \eqref{crit-sol} is unstable. 
\end{thm}
\begin{rem}
The case $\la_1=\la_2=n-1$ corresponds to the background solution
\eq{
-dt^2+ \cosh(t)^2 \cdot (g_1+g_2).
} 
If $\la_2<\la_1$ we divide into three subcases that relate to the previous cases by interchanging $\la_1$ and $\la_2$ (and $n_1$ and $n_2$) in the statements of the theorem.
\end{rem}

\begin{rem}
Relation \eqref{eq-def-rel} provides a simple interpretation for the case $\la_1,\la_2>0$. Then the first condition in the theorem concerns initial data for which the curvatures of both factors differ substantially while the last condition corresponds to initial data where both factors are of equal or almost equal Ricci-curvature. We conclude that for the case of positive curvature recollapse corresponds to a disequilibrium in curvature between the factors.
\end{rem}

\subsection{Products of non-positive factors}
For product manifolds with two non-positive factors (i.e.~flat or negatively curved), all spacetimes of the form \eqref{ansatz} have a uniform behavior -- one time direction is complete while the other one is incomplete.
\begin{thm}\label{thm-2}
Every spacetime of the form \eqref{ansatz} with $\la_1,\la_2\leq 0$ is complete in one time direction and incomplete in the other. 
\end{thm}

\begin{rem}
The previous theorem implies that within the class given by \eqref{ansatz} positive curvature of at least one factor is a necessary criterion for recollapse.

\end{rem}
\subsection{The three-dimensional case}
We have mentioned that an analogous construction of recollapsing models in 3+1 dimensions fails. If we consider the product manifold $\mathbb S^2\times \mathbb S^1$, there is, however, an interesting connection to the present study. The Nariai solution
\eq{
-dt^2+\frac{1}{3}g_{\mathbb S^2}+\cosh^2(\sqrt{3}t)dx^2,
}
where $g_{\mathbb S^2}$ denotes the metric of constant curvature $1$, is known to be unstable \cite{Be09-1,Be09-2} and shares similar features with the critical solution in Theorem \ref{thm-1}, which has one constant factor and one expanding factor (as Nariai does). However, the instability of this critical solution is obvious since it marks the border between two different regimes of initial data, in particular, it is unstable already in this class of product manifolds. We provide an elementary proof of the instability of the Nariai solution in our formalism.
\begin{thm}[cf.~\cite{Be09-1,Be09-2}]\label{thm-3}
The Nariai solution is unstable in the class of product manifolds of the type \eqref{ansatz} on $\mathbb S^1\times \mathbb S^2$, i.e.~there are arbitrarily small perturbations with an incomplete MGHD.
\end{thm}

\subsection{Discussion}
The results presented in this paper for the case where at least one factor has positive curvature concern so far only one-parameter families of initial data and their neighborhood in the initial data manifold. A complete analysis of the 3-dimensional space of initial data with this product structure with techniques similar to those used in \cite{AnHe07} is clearly desirable. However, in view of the implications for open neighborhoods in the full initial data manifold our results already establish the existence of different stable types of behavior for the Einstein-$\La$ flow in higher dimensions ($\dim M=n\geq 4$). A partial result by the authors, which concerns only even spatial dimensions has been discussed in further detail in \cite{FaKr15}.

\section{Proofs}
We proceed with proving Theorems \ref{thm-1}, \ref{thm-2} and \ref{thm-3}. We begin by presenting the Einstein equations for the product ansatz \eqref{ansatz} and recall a modified version of the Hawking singularity theorem by Andersson and Galloway \cite{AnGa04}, which applies to solutions of the Einstein equations with positive cosmological constant. This theorem allows to conclude incompleteness towards one time direction for a large set of models and is relevant for certain proofs below. We then close with the proofs of Theorems \ref{thm-1}, \ref{thm-2} and \ref{thm-3}.
\subsection{Preliminaries}
We set $\Lambda=\frac{n(n-1)}{2}$ and consider solutions of the form \eqref{ansatz}.
Then, recall $x=\log(a)$ and $y=\log(b)$, the Einsteins equations imply the following system of ODEs:
\eq{\label{eq : wp-ev}\alg{
		x''&=n-\lambda_1e^{-2x}-n_1(x')^2-n_2x'y'\\
		y''&=n-\lambda_2e^{-2y}- n_2(y')^2-n_1x'y'
	}}
and the Hamiltonian constraint reduces to
\eq{n(n-1)=n_1\lambda_1e^{-2x}+n_2\lambda_2e^{-2y}+(n_1-1)n_1(x')^2+(n_2-1)n_2(y')^2+2n_1n_2x'y',}
where $x'$ denotes the time derivative of $x$. 
The equations are a direct consequence of \cite[Proposition 2.5]{DoUe05} using the ansatz \eqref{ansatz} for the metric.
A straightforward computation shows that the mean curvature $H$, the second fundamental form $K$ and the norm of the tracefree part $\Si$ of the latter are given by
\eq{\alg{
H&=-n_1x'-n_2y',\\
K&=-x'\cdot g_1-y'\cdot g_2,\\
|\Si|^2&= \frac{n_1n_2}{n}\cdot (x'-y')^2,
}}
respectively.
\subsection{The Hawking singularity theorem with positive cosmological constant}
We cite below the generalized version of the Hawking singularity theorem from \cite{AnGa04} and derive a corollary suitable for the present purpose.
\begin{prop}[cf.~{\cite[Proposition 3.4]{AnGa04}}]\label{sing-thm}
	Let $(\overline M^{n+1},\mathbf{g})$ be a $n+1$-dimensional globally hyperbolic Lorentzian manifold and suppose 
	\eq{Ric(X,X)\geq -n |\mathbf g(X,X)|}
	for all timelike vectors $X$. Assume furthermore that there exists a hypersurface $M$ whose expansion satisfies $\theta\leq \theta_0<-n$. Then, $(\overline M,\mathbf g)$ is future geodesically incomplete.
\end{prop}

%\begin{rem}
%	In this paper we can set $K=n$ since we demand $\Lambda=\frac{(n-1)n}{2}$ and thus we get geodesic incompletenes, if $\theta\leq\theta_0<-n$. Note that the metric $-dt^2+e^{-2t}g_{\mathsf{eukl}}$ is a solution of Einsteins equation with positive cosmological constant and the expansion of all hypersurfaces is $-n$.
%\end{rem}
\begin{cor}
	Let $(\overline{M},\mathbf{g})$ be a MGHD of Einstein's equations with cosmological constant $\Lambda=\frac{n(n-1)}{2}>0$. If $\overline{M}$ contains a Cauchy hypersurface $M$ whose induced metric has scalar curvature $R(g)<R_0<0$, then $\overline{M}$ is geodesically incomplete either to the future or to the past. 
	\end{cor}
\begin{proof}
	Let $g$ be the induced metric on $M$ and as above, let $K$ be the second fundamental form of $M$ and $\Sigma=K-\frac{1}{n}H$ its tracefree part. By the constraint equations, 
	\eq{n(n-1)=R(g)-|K|^2+H^2=R(g)-|\Sigma|^2+\frac{n-1}{n}H^2\leq R_0+ \frac{n-1}{n}H^2,}
	so that 
	\eq{\theta^2=H^2\geq (n(n-1)-R_0)\frac{n}{n-1}>n^2.}
	After time reflection, if necessary, we may conclude 
	\eq{\theta\leq \theta_0<-n}
	and now future incompleteness follows from the above singularity theorem.
\end{proof}

\subsection{Proof of Theorem \ref{thm-1}}
\label{sec : 4}

\begin{proof}[Proof of Theorem \ref{thm-1}]

\noindent
We first prove Case 1. By assumptions on the $\lambda_i$, $y''(0)=n-\lambda_2<0$ and $x''(0)=n-\lambda_1>\frac{n}{n_1}>0$. Then, $x$ is monotonically increasing: If $t_0>0$ is the first time such that $x'(t_0)=0$, we get $x''(t_0)>\frac{n}{n_1}>0$ (if $\lambda_1\leq0$, we even get $x''(t_0)\geq n)$ which implies $x'(t)<0$ for $t\in (t_0-\varepsilon,t_0)$ and thus contradicts the assumption. Analogously, one shows that $y$ is monotonically decreasing.
	
	Now we show that as long as $y$ is uniformly bounded, its derivatives exist and $x$ and its derivatives cannot blow up in finite time. If $y'\leq-1$ and $y\geq Y_0$, we obtain by monotonicity, the evolution equation and $x'y'<0$ that
	\eq{y''\geq -C (y')^2}
	which after dividing by $y'$ and integration in turn implies
	\eq{\log|y'|\leq C(1-Y_0),}
	i.e. $y'$ is bounded.
	Observe that by the evolution equation, we get $x''\leq n+Cx'$ as long as $y'>-C$ which implies the assertion on $x$. 
	
	Next we show that $y$ is unbounded. If not, it must converge to a limit by monotonicity, so we have $y(\infty)=\lim_{t\to\infty} y(t)$. Then we get either a sequence of $t_i\to\infty$ such that $y''(t_i)=0$ or $\lim_{t\to\infty} y''(t)=0$. The first case cannot occur, since then
	$y''(t_i)=n-\lambda_2e^{-2y(t_0)}\geq \frac{n}{n_1}>0$. In the latter case, we necessarily get $\lim x'(t)=\infty$ by the evolution of $y$. On the other hand, by the constraint equation,
	\eq{n_1x''-n_2y''=n(n_1-n_2)-n_1\lambda_1e^{-2x}+n_2\lambda_2e^{-2y}-n_1^2(x')^2+n_2^2(y')^2\\
		\leq C-n_1^2(x')^2+n_2^2(y')^2}
	From the constraint and adding the evolution equations, we also easily obtain
	\eq{n_1(x''+(x')^2)+n_2(y''+(y')^2)=n,} 
	from which we get
	\eq{-n_1(x')^2-2n_2y''+n-n_2(y')^2\leq C-n_1^2(x')^2+n_2^2(y')^2}
	and therefore, since $n_1>1$ (if $n_1=1$, the above bound implies $\lambda_1<0$ which in turn contradicts $n_1=1$), we get a bound on $x'$ if $y'$ and $y''$ are bounded. This yields the contradiction.
	
	By monotonicity of $x$ and $y$ and the unboundedness of $y$ from below there is a constant $C>0$ and a time $t_1>0$ such that the differential inequality 
	\eq{
		n_1x''+n_2y''\leq-C-(n_1x'+n_2y')^2
	}
	holds for all $t\geq t_1$. By ODE comparison,
	$z(t)=n_1x'(t)+n_2y'(t)\leq -\sqrt{C}\tan(D+\sqrt{C}(t-t_1))$ for $t\geq t_1$ and some $D\in\R$. Therefore, $z\to-\infty$ as $t$ approaches some $t_+$. As $x'>0$, $y'\to-\infty$. We also have $y(t)\leq\frac{1}{n_2}\int_0^{t}z(s)ds\to-\infty$ as $t\to t_+$, which implies that $y$ diverges.
	In addition, because
	\eq{x''=n-\lambda_1e^{-2x}-zx'\geq -zx',}
	we obtain $x(t)\geq \int_0^t\exp(-\int_0^sz(r)dr)ds\to \infty$ as $t\to t_+$.

It remains to show that the Kretschmann scalar is unbounded at the singularity:
 	 Let $0$ be the index referring to the $t$-coordinate and $i,j$ and $a,b$ be indices referring to coordinates on $M_1,M_2$ respectively. One can check that the coefficients with an odd number of zeros vanish and therefore, the Kretschmann scalar satisfies
 	\eq{ |\tilde{R}|^2=\tilde{R}_{\mu\nu\lambda\sigma}\tilde{R}^{\mu\nu\lambda\sigma}\geq 0}
 	because all summands are nonnegative. Now it can be checked that
 	\eq{ \tilde{R}_{0ij0}=-\frac{\partial^2_t a}{a} \tilde{g}_{ij},\qquad \tilde{R}_{0ab0}=-\frac{\partial^2_t b}{b} \tilde{g}_{ab},}
 	so that
 	\eq{|\tilde{R}|^2\geq \left(n_1\frac{\partial^2_t a}{a}\right)^2+n_2\left(\frac{\partial^2_t b}{b}\right)^2=n_1(x''+(x')^2)^2+n_2(y''+(y')^2)^2}
 	where as above $x=\log(a)$, $y=\log(b)$.
 	Now, since $x\to\infty$ as $t\to t_+$,
 we also get $x''+(x')^2\to\infty$, which shows that the Kretschmann scalar must blow up. This finishes the proof in Case 1. In Case 2, it is easy to see that the metric given in the statement of the theorem is the MGHD of the given initial data.
	%\begin{lem}\label{equilibrium}
	%	For equilibrium initial data initial data,
	%	the solution $(x(t),y(t))$ of the system
	%	\eqref{eq : wp-ev} with initial data $x(0)=y(0)=x'(0)=y'(0)=0$ exists for all $t>0$. Moreover, if $s< \frac{n-1}{n-2}$, we have a limit
	%	\eq{\label{limit x-y}C_s=\lim_{t\to\infty}(x(t)-y(t)).}
	%\end{lem}
	%
		Let us now consider Case 3. By assumption,
		\eq{
			y''(0)=n-\lambda_2\in(0,1)
		}
		and
		\eq{
			x''(0)=n-\lambda_1\in (1,\frac{n}{n_1}).
		}
		Both $x$ and $y$ are strictly monotonically increasing. We have $x'(t)>0$ for small $t$. Let $t_0$ be the first time, where $x'(t_0)=0$. Then  we obtain
		\eq{
			x''(t_0)=n-\lambda_1e^{-2x(t_0)}>0,
		}
		as long as $y'(t_0)$ exists. Thus $x'<0$ on $(t_0-\varepsilon,t_0)$, which causes the contradiction.
		Analogously one shows that $y$ is strictly monotonically increasing. From the evolution equations and monotonicity we deduce
		\eq{
			n^2-n_1\lambda_1-n_2\lambda_2-(n_1x'+n_2y')^2\leq n_1x''+n_2y''\leq n^2-(n_1x'+n_2y')^2.
		}
		The solution of the corresponding ODE, $z'=n^2-z^2$,
	 is $z(t)=n\tanh(nt+c)$, where $c\in \R$. By the initial conditions this implies $0<n_1x'+n_2y'<n$ for all $t>0$. Due to the positivity of $x'$ and $y'$ these statements hold for $x'$ and $y'$ individually.
		In particular, $x$ and $y$ exist for all times, which finishes the proof of Case 3.

Finally, recollapse of MGHDs of small perturbations of the initial data induced by the models follows by the Cauchy stability of the Einstein equations in combination with Theorem \ref{sing-thm}. As the background solutions have Cauchy surfaces of strictly positive and strictly negative mean curvature, this also holds for the MGHDs of small perturbations. Then the singularity theorem yields incompleteness. This idea of proof is explained in further detail originally in \cite{Ri09}. Stability in the expanding direction follows from the results in \cite{Ri08}.		
		\end{proof}
	\begin{rem}\label{expansionrate}	With some more effort, we can determine the expansion rate of the warping factors of the  solutions appearing in Case 3:	
		By the above, we have $n_1x'(t)+n_2y'(t)>C_1>0$ for all $t\geq t_1$ and $n_1x(t)+n_2y(t)>C_2t$ for all $t\geq t_1$ and $C_2>0$. Using this, we can improve the above estimate. Denoting $z=n_1x'+n_2y'$ as above, we get
		\eq{
			n^2-Ce^{-Ct}-z^2\leq z'\leq n^2-z^2.
		}	
		from which we obtain the uniform bounds $0\geq z-n\geq -Ce^{-Ct}$. After integrating,
		\eq{\label{z}0\geq n_1x+n_2y-nt\geq -C.}
		Similarly, we obtain for $w:=n_1x'-n_2y'$ the estimate
		\eq{ n(n_1-n_2)-Ce^{-C t} -zw \leq w'\leq  n(n_1-n_2)+Ce^{-C t} -zw}
		which implies because of the estimate on $n-z$ that $|w-(n_1-n_2)|\leq Ce^{-Ct}$. Again by integrating,
		\eq{\label{w}{C_1\geq n_1x-n_2y-(n_1-n_2)t\geq C_2.}}	
		Finally, we obtain from \eqref{z} and \eqref{w} that $\sup \left\{|x-t|,|y-t|\right\}\leq C$. It is immediate that the warping functions satisfy $C_1 e^t\leq a(t) \leq C_2e^t$ and $C_1 e^t\leq b(t) \leq C_2e^t$ for $t\geq0$.
		\end{rem}
		
%		In the boundary case $s=\tfrac{n-1}{n-2}$, $y''(0)=0$. Thus, $y\equiv 0$ and the system reduces to the initial value problem
%		\eq{x''=n-(n-2)e^{-2x}-\frac{n}{2}(x')^2,\qquad x(0)=x'(0)=0.}
%		By similar arguments as above, one shows that $x$ is strictly monotonically increasing. An immediate implication is $x''<n$, which implies that $x$ grows at most quadratically. Therefore, it it exists for all time.

\subsection{Proof of Theorem \ref{thm-2}}

\begin{proof}[Proof of Theorem \ref{thm-2}]

If $\lambda_1,\lambda_2\leq 0$, we are going to show that all corresponding solutions are geodesically incomplete in one direction and geodesically complete in the other.
	Let us first remark that the case where $R(g)\equiv |\Sigma|\equiv 0$ corresponds to the solution $-dt^2+e^{2t}g$ which is geodesically incomplete in the past and geodesically complete in the future. Therefore we assume $R(g)-|\Sigma|^2\leq R_0<0$ from now on.
		 By the constraint equation, we have 
		\eq{n(n-1)=R(g)-|K|^2+H^2=R(g)-|\Sigma|^2+\frac{n-1}{n}H^2\leq R_0+ \frac{n-1}{n}H^2,}
		As above, we get $|\theta|\geq \theta_0>n$. If $\theta\leq-\theta_0<-n$, we obtain geodesic incompleteness by Proposition \ref{sing-thm}. If $\theta\geq\theta_0>n$, we have the differential inequality
			\eq{ \partial_t \theta\leq n-\frac{1}{n}\theta^2}
	so that we have the upper bound
			\eq{\theta \leq n\cdot\frac{e^{2t}-A_0}{e^{2t}+A_0}= \frac{n}{1-|A_0|e^{-2t}}(1+|A_0|e^{-2t})\leq n+n|A_0|e^{-2t}}
		where $A_0=\frac{n-\theta_0}{n+\theta_0}$. Again by the constraint equation,
		\eq{0\leq-R(g)+|\Sigma|^2= \frac{n-1}{n}(\theta^2-n^2)\leq Ce^{-2t}. }
		In particular, since $R(g)=n_1\lambda_1e^{-2x}+n_2\lambda_2 e^{-2y}$, both factors $x,y$ grow linearly, in particular, they do not converge to $-\infty$. Thus, they are geodesically complete.\end{proof}
		\begin{rem}
			Since $\theta\to n$ and $ |\Sigma|^2\to 0$ exponentially, one concludes that $x',y'\to 1$ exponentially and $\sup \left\{|x-t|,|y-t|\right\}\leq C$ as in Remark \ref{expansionrate}.
		\end{rem}
		\begin{rem}
			We expect that the generic type of singularity that occurs in the geodesically incomplete direction is of Kasner type, i.e. one of the factors goes to $\infty$ while the other goes to $-\infty$ as $t\to t_+$. This also implies a blow-up of the Kretschmann scalar as in the proof of Theorem \ref{thm-1}.
		\end{rem}
\subsection{Proof of Theorem \ref{thm-3}}
\begin{proof}[Proof of Theorem \ref{thm-3}]
We have the data $n_1=2$, $n_2=1$, $\lambda_1=\lambda$, $\lambda_2=0$, so the evolution equations are
\eq{\alg{x''&=3-\lambda e^{-2x}-2(x')^2-x'y'\\
	y''&=3-(y')^2-2x'y'}}
and the constraint equation reduce to
\eq{6=2\lambda e^{-2x}+2(x')^2+4x'y'.}
We demand $y'(0)=0$ such that we obtain a one-parameter family of solutions by prescribing $\lambda$ and $x'(0)$ such that
\eq{3=\lambda+(x'(0))^2}
by the constraint. Recall that we generally require $x(0)=0$. For our purposes, we additionally assume $\lambda>0$ throughout.
 By using the constraint, the equation on $x$ decouples as
\eq{\alg{x''&=3-\lambda e^{-2x}-2(x')^2-x'y'=\frac{1}{2}(3-\lambda e^{-2x}-3(x')^2).}}
Now we show that the MGHD is future incomplete, if $\lambda<3$ and $x'(0)<0$. At first,
\eq{x''(0)=3-\lambda-2x'(0)^2=(3-\lambda-(x'(0))^2)-x'(0)^2<0,}
which shows that both $x$ and $x'$ are monotonically decreasing for small time. In addition, as long as $x$ and $x'$ are monotonically decreasing, we have
\eq{x''=\frac{1}{2}(3-\lambda e^{-2x}-3(x')^2)\leq\frac{1}{2}(3-\lambda-2(x'(0)))\leq \frac{1}{2}x''(0)<0,}
which shows that this monotonicity property holds for all time and $x$ and $x'$ are unbounded. In particular, for large time, we get
\eq{x''\leq -\frac{3}{2}(x')^2. }
This shows that $x'\to -\infty$ as $t$ approaches some $t_+$, since the comparison ODE is solved by $z(t)=\frac{2}{c+3t}$. By integrating, one obtains the estimate
\eq{x(t)\leq C_1+\log(C_2-\frac{3}{2}t)}
for some constants  $C_1,C_2>0$ which shows that also $x\to-\infty$ as $t\to t_+$.
 	\end{proof}
 	\begin{rem}We are also able to show that this singularity is of Kasner type, i.e. $y(t)\to\infty$ as $t\to t_+$. At first note that $y$ is monotonically increasing for small time. In fact, it is monotonically increasing for all $t>0$ since $y''(t_0)=3>0$ as long as $y'(t_0)=0$. But then we have an estimate 
\eq{y'' \geq 3-(y')^2+2 |z|y'}
and since $z\to\infty$ as $t\to t_+$ this shows that $y'\to\infty$ as $t\to t_+$ and because the integral of $z$ diverges, we get that also $y$ diverges as $t\to t_+$.
\end{rem}
\begin{rem}
A similar analysis as in the proof of Theorem \ref{thm-1} shows that every MGHD is complete in the other time-direction and arguments as in Remark \ref{expansionrate} imply exponential growth of both warping functions in the past.
\end{rem}
%
%Let us now consider the case where $x'(0)>0$ which can also be seen as a time reflection. In this case $x$ is monotonically increasing for all time: If $t_0$ as the first time such that $x'(t_0)=0$, we get
%\eq{x''(t_0)=\frac{1}{2}(3-\lambda e^{-2x(t_0)})\geq \frac{1}{2}(3-\lambda)>0.}
%on the other hand, $x''\leq \frac{3}{2}$ which implies that $x$ is bounded for finite time. Now, one can show by analogous arguments that $y$ is monotonically increasing. We also get the bound $y''\leq 3$ quite immediately which shows that both $x$ and $y$ exist for all time.
%
%Now one can use exactly the same arguments as in section ??? to show $\sup \left\{|x-t|,|y-t|\right\}\leq C$, i.e. both factors grow exponentially.
%	
	%%%

\vspace{0.1cm}
\noindent
\textsc{David Fajman\\
Faculty of Physics, University of Vienna,\\
Boltzmanngasse 5, 1090 Vienna, Austria}\\ 
\vspace{-0.4cm}\\
\texttt{David.Fajman@univie.ac.at}\\

%\vspace{0.1cm}
\noindent
\textsc{Klaus Kr\"oncke\\
Faculty of Mathematics, University of Hamburg,\\
Bundesstrasse 55, 20241 Hamburg, Germany}\\ 
\vspace{-0.4cm}\\
\texttt{klaus.kroencke@uni-hamburg.de}\\


\begin{thebibliography}{DWW07}

% this bibliography is generated by alphadin.bst [8.2] from 2005-12-21

\providecommand{\url}[1]{\texttt{#1}}
\expandafter\ifx\csname urlstyle\endcsname\relax
  \providecommand{\doi}[1]{doi: #1}\else
  \providecommand{\doi}{doi: \begingroup \urlstyle{rm}\Url}\fi

\bibitem[An01]{An01}
\textsc{Anderson}, M.~T.~:
\newblock {On Long-Time evolution in General Relativity and Geometrization of 3-Manifolds}
\newblock {In: }\emph{Commun.~Math.~Phys.~} \textbf{222} (2001), 533--567

\bibitem[AnGa04]{AnGa04}
\textsc{Andersson}, L.~ ; \textsc{Galloway}, G.~:
\newblock {dS/CFT and spacetime topology}
\newblock {In: }\emph{Adv.~Theor.~Math.~Phys.} \textbf{6} (2002), 307--327

\bibitem[AnHe07]{AnHe07}
\textsc{Andersson}, L.~ ; \textsc{Heinzle}, J.~M.~:
\newblock {Eternal acceleration from M-theory}
\newblock {In: }\emph{Adv.~Theor.~Math.~Phys.} \textbf{11} (2007), 371--398




\bibitem[AnMo11]{AnMo11}
\textsc{Andersson}, L.~; \textsc{Moncrief}, V.~:
\newblock {Einstein spaces as attractors for the Einstein flow.}
\newblock {In: }\emph{J. Differ. Geom.} \textbf{89} (2011), no. 1, 1--47



\bibitem[BaGaTi86]{BaGaTi86}
\textsc{Barrow}, J.~D. ; \textsc{Galloway}, G.~J. \textsc{Tipler}, F.~J.~;
\newblock{The closed-universe recollapse conjecture}
\newblock{In: }\emph{Month.~Not.~R.~Astr.~Soc.}, \textbf{4} (1986), 835-844



%\bibitem[Be08]{Bes08}
%\textsc{Besse}, Arthur~L.:
%\newblock \emph{{Einstein manifolds. Reprint of the 1987 edition.}}
%\newblock {Berlin: Springer}, 2008

\bibitem[Be09-1]{Be09-1}
\textsc{Beyer}, F.:
\newblock \emph{{Non-genericity of the Nariai solutions: I. Asymptotics and spatially homogeneous perturbations}}
\newblock \emph{Class.~Quant.~Grav.}, \textbf{26}, (2009) 235015

\bibitem[Be09-2]{Be09-2}
\textsc{Beyer}, F.:
\newblock \emph{{Non-genericity of the Nariai solutions: II. Investigtions within the Gowdy class}},
\newblock \emph{Class.~Quant.~Grav.}, \textbf{26}, (2009) 235016

\bibitem[BuRe96]{BuRe96}
\textsc{Burnett}, G.~A.; \textsc{Rendall}, A.~D.:
\newblock \emph{{Existence of maximal hypersurfaces in some spherically symmetric spacetimes}},
\newblock \emph{Class.~Quant.~Grav.}, \textbf{13}, (1996) 111--123

%\bibitem[ChCo02]{CBC02}
%\textsc{Choquet-Bruhat}, Y.~; \textsc{Cotsakis}, S.~:
%\newblock {Global hyperbolicity and completeness.}
%\newblock \emph{J.~Geom.~Phys.} \textbf{43} (2002), 345--350

\bibitem[DoUe05]{DoUe05}
\textsc{Dobarro}, F.~ ;  \textsc{\"Unal}, B.~: 
\newblock{Curvature of multiply warped products.}
\newblock \emph{J.~Geom.~Phys.} \textbf{55} (2005): 75--106.

\bibitem[FaKr15]{FaKr15}
\textsc{Fajman}, D.~; \textsc{Kr\"oncke}, K.:
\newblock{Stable fixed points of the Einstein flow with positive cosmological constant},
\newblock \emph{arXiv:1504.00687} (2015)



\bibitem[He02]{He02}
\textsc{Henkel}, O.:
\newblock {Global prescribed mean curvature foliations in cosmological space-times. I, II}
\newblock {In: }\emph{J.~Math.~Phys.} \textbf{43} (2002), 2439--2465, 2466--2485

\bibitem[Re96]{Re96}
\textsc{Rendall}, Alan D.:
\newblock {Constant mean curvature foliations}
\newblock {In: }\emph{Helv.~Phys.~Acta} \textbf{69} (1996), 490--500


\bibitem[Ri08]{Ri08}{
   \textsc{Ringstr{\"o}m, H.~},
  \newblock{Future stability of the Einstein-non-linear scalar field system},
   \newblock {In: }\emph{Invent. Math.},
   \textbf{173},
   (2008), {123--208}}




\bibitem[Ri09]{Ri09}
\textsc{Ringström}, H.~
\newblock \emph{The Cauchy problem in general relativity.}
\newblock ESI Lectures in Mathematics and Physics, Zürich, European Mathematical Society Publishing House



\end{thebibliography}
\end{document}